\documentclass[a4paper]{article}
\usepackage{fullpage}
\usepackage{amsmath,amssymb,amsthm}
\usepackage{enumerate}
\usepackage{tikz}
\usetikzlibrary{matrix,calc,automata,graphs,graphs.standard}
\usepackage{fancybox}
\usepackage{fancybox,ascmac}
\newtheorem{theorem}{Theorem}
\newtheorem{lemma}{Lemma}
\newcommand{\trail}[1]{\mathcal{T}_{#1}}
\title{Ramsey Numbers of Trails}

\author{Masatoshi Osumi\thanks{The University of Electro-Communications}}
\date{\today}
\begin{document}
\maketitle

\begin{abstract}
We initiate the study of Ramsey numbers of trails.
Let $k \geq 2$ be a positive integer.
The Ramsey number of trails with $k$ vertices is defined as the
the smallest number $n$ such that
for every graph $H$ with $n$ vertices,
$H$ or the complete $\overline{H}$ contains a trail with $k$ vertices.
We prove that the Ramsey number of trails with $k$ vertices is at most
$k$ and at least $2\sqrt{k}+\Theta(1)$.
This improves the trivial upper bound of $\lfloor 3k/2\rfloor -1$.
\end{abstract}

\section{Introduction}

Ramsey theory is one of the topics in discrete mathematics that has been
studied over the years \cite{GRS,KR}.
For graphs, the Ramsey number was first studied for complete graphs, and
later it was studied for other classes of graphs such as paths, cycles, and
trees.
For graphs $G$ and $G'$, the Ramsey number of the pair $(G, G')$ is the 
smallest number $n$ such that 
for every graph $H$ with $n$ vertices, $H$ contains a copy of $G$ or
the complement $\overline{H}$ contains a copy of $G'$.
It is known that for every pair $(G, G')$ of finite graphs, the Ramsey number
of $(G, G')$ exists, 
and the determination of the Ramsey number is the ultimate goal.
However, even for complete graphs, the exact Ramsey number is not known:
When $G=G'=K_5$ we only know that the Ramsey number lies between
$43$ and $48$ \cite{survey}.

In this paper, we initiate the study of Ramsey numbers for trails.
Unlike paths, trails may have a repetition of vertices.
To study the Ramsey number of trails, we first fix the number of
vertices in a trail.
Let $k$ and $\ell$ be integers.
Then, the Ramsey number of trails with $k$ vertices and $\ell$ vertices
is defined as
the smallest number $n$ such that
for every graph $H$ with $n$ vertices,
$H$ contains a trail with $k$ vertices or
$\overline{H}$ contains a trail with $\ell$ vertices.

The ultimate goal is to determine the Ramsey number of trails.
Unfortunately, we are unable to provide a definite answer.
Nonetheless, we give a progress toward the ultimate goal.
We concentrate on the diagonal case, i.e., the case where $k=\ell$.
Our main theorems give an improved upper bound of $k$, and also a
lower bound of roughly $2\sqrt{k}$.
We note here that a trivial upper bound is $\lfloor 3k/2\rfloor -1$,
which will be sketched in the next section.

\section{Preliminaries}

In this paper, all graphs are finite, simple and undirected.
A \emph{graph} $G$ is defined as a pair $(V,E)$ of 
a finite set $V$ and $E \subseteq \{\{u,v\} \mid u,v \in V, u\neq v\}$, where
$V$ is the set of \emph{vertices} of $G$ and
$E$ is the set of \emph{edges} of $G$.
The \emph{degree} of a vertex $v \in V$ is the number of edges incident to
$v$, i.e., $|\{e \in E \mid v \in e\}|$.

A graph $G'=(V',E')$ is a \emph{subgraph} of a graph $G=(V,E)$ if
$V' \subseteq V$, $E' \subseteq E$ and $u,v\in V'$ for every $e =\{u,v\}\in E'$.
For a graph $G=(V,E)$, the \emph{complement} of $G$, denoted by $\overline{G}$,
is a graph with vertex set $V$ and edge set $\overline{E} = \{\{u,v\} \mid u,v \in V,u\neq v, \{u,v\} \notin E\}$.
Namely, $\overline{G} = (V,\overline{E})$.
A pair $(G,H)$ of graphs is called \emph{complementary} if
$H=\overline{G}$.

A graph is A \emph{complete} if each pair of vertices is
joined by an edge. The complete graph with $n$ vertices is denote by
$K_n$.
A graph $P=(V,E)$ is a \emph{path} if
$V = \{v_1,v_2,\dots,v_n\}$, and
$E = \{\{v_i,v_{i+1}\}\mid i\in\{1,2,\dots,n-1\}\}$.
The path with $n$ vertices is denote by $P_n$.

A \emph{walk} is a sequence $v_1e_1v_2\dots e_{k-1}v_k$ of vertices $v_i$ and edges $e_i$ such that for $1\leq i \leq k$, the edge $e_i=\{v_i, v_{i+1}\}$.
Here, $k$ is the number of vertices of the walk, 
$v_1$ and $v_k$ are called endpoints. 
A \emph{trail} is a walk in which all the edges are different from each other.
A trail that satisfies $v_1=v_k$ is called a \emph{circuit}.
A graph is \emph{connected} if it has a trail from any vertex to any other vertex.

Let $G$ be a connected graph.
An \emph{Eulerian circuit} of $G$ is a
circuit of $G$ that passes every edge exactly once. 
If $G$ has an Eulerian circuit, then $G$ is called \emph{Eulerian}.
An \emph{Eulerian trail} of $G$ is a trail of $G$ that passes every edge 
exactly once.
If $G$ has an Eulerian trail but no Eulerian circuit, 
then $G$ is called a \emph{semi-Eulerian}.
It is well-known and easy to prove that 
a connected graph $G$ is Eulerian if and only if
the degree of every vertex of $G$ is even, and
$G$ is semi-Eulerian if and only if
the number of odd-degree vertices is two.

For $k\geq 1$, 
we denote by $\trail{k}$ the set of connected graphs that have an Eulerian circuit or an Eulerian trail with $k$ vertices.
Note that in our definitions, vertices in trails and circuits are counted multiple times if they are passed multiple times.
Therefore, some graphs in $\trail{k}$ may have less than $k$ vertices.

\begin{figure}[t]
\centering
\begin{tikzpicture}[]
  \fill [white] (0,0);
  \fill [black] (0,0.5) circle (2pt);
  \fill [black] (0.5,0.5) circle (2pt);
  \fill [black] (1,0.5) circle (2pt);
  \fill [black] (1.5,0.5) circle (2pt);
  \draw[] (0,0.5) -- (0.5,0.5) -- (1,0.5) -- (1.5,0.5);
  
  \fill [black] (2.5,0) circle (2pt);
  \fill [black] (3.5,0) circle (2pt);
  \fill [black] (3,0.865) circle (2pt);
  \draw (2.5,0) -- (3.5,0) -- (3,0.865) --cycle;
\end{tikzpicture}
  
\caption{Graphs in $\trail{4}$. Note that the right graph has only three vertices, but it has a trail with four vertices.}
\label{trail}
\end{figure}
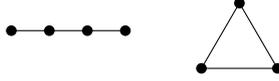 

Let $\mathcal{C}$ and $\mathcal{C}'$ be two graph classes, i.e., 
possibly infinite sets of graphs.
Then, the \emph{Ramsey number} of $\mathcal{C}$ and $\mathcal{C}'$ is
the smallest number $n$ such that
for every graph $H$ with $n$ vertices, 
$H$ contains a graph in $\mathcal{C}$ or 
$\overline{H}$ contains a graph in $\mathcal{C}'$.
The Ramsey number of $\mathcal{C}$ and $\mathcal{C}'$ is denoted by 
$R(\mathcal{C}, \mathcal{C}')$.
If $\mathcal{C}$ and $\mathcal{C}'$ are singletons (i.e., contain only one graph as $\mathcal{C}=\{G\}$ and $\mathcal{C}'=\{G'\}$), then the Ramsey number of
$\mathcal{C}$ and $\mathcal{C}'$ is denoted by
$R(G,G')$.



Gerencs\'{e}r and  Gy\'{a}rf\'{a}s \cite{GG} determined the exact value of
the Ramsey number of paths, as in the following theorem.
\begin{lemma}[\cite{GG}]
Let $k\geq \ell\geq 2$.
Then, 
$R(P_k,P_\ell)=k+\lfloor \ell/2 \rfloor-1$.
\end{lemma}

Since $P_k$ belongs to $\trail{k}$,
$R(\trail{k}, \trail{k}) \leq R(P_k, P_k)=k+\lfloor k/2 \rfloor-1=\lfloor 3k/2  \rfloor-1$.
This is the trivial upper bound mentioned in the previous section.

To gain the first impression, we have conducted a computer search of
the Ramsey number $R(\trail{k}, \trail{k})$ for small values of $k$.
This has been performed with the following procedure.
For $2\leq n\leq 7$, 
we generate all graphs $G$ with $n$ vertices.
For each such $G$, we calculate $\mathrm{t}(G)$, which is defined as the
number of vertices in the longest trail in $G$ or $\overline{G}$.
Then, we determine $\mathrm{value}(n)$, which is defined as the
minimum value of $\mathrm{t}(G)$ for all $G$ with $n$ vertices.
If $k$ satisfies $\mathrm{value}(n-1)< k \leq \mathrm{value}(n)$, then 
we know that $R(\trail{k}, \trail{k})$ is equal to $n$.

The result of the computer search is summarized in \tablename~\ref{table}.
We may observe that the upper bound of $\lfloor 3k/2 \rfloor-1$
should be improved.

\begin{table}[t]
  \centering
  \caption{The values of $R(\trail{k},\trail{k})$.}
  \begin{tabular}{c|cccccccccc}
    $k$ & 2 & 3 & 4 & 5 & 6 & 7 & 8 & 9 & 10 \\ 
    \hline
    $R(\trail{k}, \trail{k})$ & 2 & 3 & 4 & 5 & 6 & 6 & 6 & 7 & 7 \\
  \end{tabular}
  \label{table}
\end{table}

\section{Main Theorem: Lower Bound}

We begin with a lower bound of $R(\trail{k}, \trail{k})$.
\begin{theorem}
  \label{thm:lb}
  Let $k$ be a positive integer. Then, 
    \begin{align*}
    R(\trail{k},\trail{k}) \geq \begin{cases}
        k & \text{if } k \leq 6,  \\
        \displaystyle \left\lceil \frac{1+\sqrt{16k-7}}{2} \right\rceil  & \text{if }k\geq 7.\\
      \end{cases}
    \end{align*}
\end{theorem}

The rest of the section is devoted to the proof of Theorem \ref{thm:lb}.
We first consider the case when $k\leq 6$.

Let $k=2$.
Then, there is no graph of $\trail{2}$ in the complete graph $K_1$. Therefore, $R(\trail{2},\trail{2})\geq 2$.

Let $k=3$.
Then, the complete graph $K_2$ has only one edge.
So, there is no graph of $\trail{3}$ in $K_2$. Thus, $R(\trail{3},\trail{3})\geq 3$.

Let $k=4$.
Consider the complementary pair of graphs with three vertices as shown in \figurename~\ref{k4}.
Since those two graphs have at most two edges, no element of $\trail{4}$ 
is contained in either graph. 
Therefore, $R(\trail{4},\trail{4})\geq 4$. 

\begin{figure}[t]
  \centering
  \begin{tikzpicture}[]
    \fill [black] (0,0) circle (2pt);
    \fill [black] (1,0) circle (2pt);
    \fill [black] (0.5,0.865) circle (2pt);
    \draw[] (0,0) -- (0.5,0.865);
    
    \draw[] (1.5,-0.5) -- (1.5,1.5);
    
    \fill [black] (2,0) circle (2pt);
    \fill [black] (3,0) circle (2pt);
    \fill [black] (2.5,0.865) circle (2pt);
    \draw[] (2,0) -- (3,0) -- (2.5,0.865);
  \end{tikzpicture}
  
  
  \caption{A complementary pair of graphs that contain no elements of $\trail{4}$.}
  \label{k4}
\end{figure}
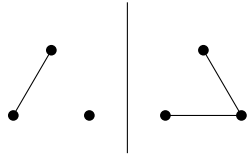

Let $k=5$.
Consider the complementary pair of graphs with four vertices as shown in \figurename~\ref{k5}.
Since those two graphs have three edges, 
no element of $\trail{5}$ is contained in either graph. 
Hence, $R(\trail{5},\trail{5})\geq 5$. 
\begin{figure}[t]
  \centering
  \begin{tikzpicture}[]
    \fill [black] (0,0) circle (2pt);
    \fill [black] (1,0) circle (2pt);
    \fill [black] (1,1) circle (2pt);
    \fill [black] (0,1) circle (2pt);
    \draw[] (0,0) -- (0,1) -- (1,1);
    \draw[] (0,1) -- (1,0);
    
    \draw[] (1.5,-.5) -- (1.5,1.5);
    
    \fill [black] (2,0) circle (2pt);
    \fill [black] (3,0) circle (2pt);
    \fill [black] (3,1) circle (2pt);
    \fill [black] (2,1) circle (2pt);
    \draw[] (2,0) -- (3,0) -- (3,1) -- cycle;
  \end{tikzpicture}
    
    
  \caption{A complementary pair of graphs that contain no elements of $\trail{5}$.}
  \label{k5}
\end{figure}
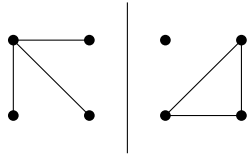

Let $k=6$.
Consider the complementary pair of graphs with five vertices as shown in \figurename~\ref{k6}.
Those graphs have five edges. 
Hence, for an element of $\trail{6}$ to be contained in either of the two graphs, one of the two graphs must be Eulerian or semi-Eulerian. 
However, each graph has four odd-degree vertices.
Thus, these graphs are neither Eulerian nor semi-Eulerian, and have no elements of $\trail{6}$.

\begin{figure}[t]
  \centering
  \begin{tikzpicture}[]
    \coordinate [rotate=18] (name) at (72:1);
    \coordinate [rotate=18] (b) at (144:1);
    \coordinate [rotate=18] (c) at (216:1);
    \coordinate [rotate=18] (d) at (288:1);
    \coordinate [rotate=18] (e) at (0:1);
    
    \fill [black,rotate=18]  (72:1) circle (2pt);
    \fill [black,rotate=18] (144:1) circle (2pt);
    \fill [black,rotate=18] (216:1) circle (2pt);
    \fill [black,rotate=18] (288:1) circle (2pt);
    \fill [black,rotate=18] (0:1) circle (2pt);
    \draw[rotate=18] (0:1) -- (72:1) -- (144:1) -- cycle;
    \draw[rotate=18] (0:1) -- (288:1);
    \draw[rotate=18] (144:1) -- (216:1);
    
    \draw[] (1.5,-1) -- (1.5,1.5);

    \fill [black] (3,0) + (18:1) circle (2pt);
    \fill [black] (3,0) + (90:1) circle (2pt);
    \fill [black] (3,0) + (162:1) circle (2pt);
    \fill [black] (3,0) + (234:1) circle (2pt);
    \fill [black] (3,0) + (306:1) circle (2pt);
    
    \draw (3,0) +(90:1) -- +(234:1) -- +(306:1) --cycle;
    \draw (3,0) + (234:1) -- +(18:1);
    \draw (3,0) + (306:1) -- + (162:1);
  \end{tikzpicture}
    
    

    
  \caption{A complementary pair of graphs that contain no elements of $\trail{6}$.}
  \label{k6}
\end{figure}
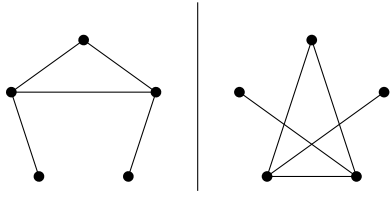


Next, we consider the case where $k \geq 7$.
To complete the proof, we use the following two lemmas.

\begin{lemma}
  \label{lem:m-complete}
  The number of vertices of a complete graph with $m\geq 0$ edges is $(1+\sqrt{1+8m})/2$.
\end{lemma}
\begin{proof}
  Let $n\geq 1$ be the number of vertices of a complete graph with $m$ edges. In this case, $m = n(n-1)/2$. Solving for $n \geq 1$, we have $n=(1+ \sqrt{1+8m})/2$.
\end{proof}

\begin{lemma}
  \label{lem:lb-k7}
  Let $k\geq 7$ and let $n$ be such that the complete graph $K_n=(V,E)$ 
  has at most $2k-2$ edges.
  Then, there exists a subgraph $G=(V,E_1)$ of $K_n$ 
  such that $G$ and $\overline{G}$ have no element of $\trail{k}$.
\end{lemma}

Before proving Lemma \ref{lem:lb-k7}, we finish the proof of Theorem
\ref{thm:lb} using Lemmas \ref{lem:m-complete} and \ref{lem:lb-k7}.

\begin{proof}[Proof of Theorem \ref{thm:lb} when $k\geq 7$.]
Let $k\geq 7$ and let $n$ be the number of vertices of a complete graph with 
 at most $2k-1$ edges.
Then, $R(\trail{k},\trail{k})\geq n$ from Lemma \ref{lem:lb-k7}.
Therefore, by Lemma \ref{lem:m-complete}
\[
  R(\trail{k},\trail{k})\geq 
  \left\lceil \frac{1+\sqrt{1+8(2k-1)}}{2} \right\rceil
  =
  \left\lceil \frac{1+\sqrt{16k-7}}{2} \right\rceil.
\qedhere
\]
\end{proof}

Thus, it suffices to prove Lemma \ref{lem:lb-k7}.
\begin{proof}[Proof of Lemma \ref{lem:lb-k7}.]
We distinguish the cases $|E| \leq 2k-4$, $|E|=2k-3$ and $|E|=2k-2$.

\paragraph{Case 1: $|E|\leq 2k-4$.}
Choose $G=(V,E_1)$ as any subgraph with $|E_1|=\left\lceil |E|/2\right\rceil$.
Then, since $|E_1|= \left\lceil |E|/2 \right\rceil \leq k-2,|E-E_1|= \left\lfloor |E|/2\right\rfloor\leq k-2$, 
it follows that $G$ and $\overline{G}$ have no element of $\trail{k}$.

\paragraph{Case 2: $|E|=2k-3$.}
Consider a complete graph $K_n$ with $2k-3$ edges.
Then, 
\[
n = \frac{1+ \sqrt{1+8(2k-3)}}{2}=\frac{1+\sqrt{16k-23}}{2} > 5
\]
from Lemma \ref{lem:m-complete}.
Let $v_1, \dots, v_n$ be the vertices of $K_n$, 
and let $E_c=\{\{v_i,v_{i+1}\}\mid i=1,2,\dots,n-1\}\cup \{v_n,v_1\}$.
Then, the graph $C=(V,E_c)$ is a cycle contained in $K_n$.
We construct an Eulerian or a semi-Eulerian graph $S$ with $k+1$ edges 
that contains $C$.

Before constructing such $S$, 
we observe that this is enough for our purpose.
Since $S$ contains $C$ and $n\geq 6$, 
there exist two edges $e_1, e_2$ of $C$ such that
$S' = S - \{e_1, e_2\}$ has exactly four odd-degree vertices.
Thus, $S'$ is neither Eulerian nor semi-Eulerian.
Since $S'$ has only $k-1$ edges, $S'$ include no element of $\trail{k}$.
Further, $\overline{S'}$ has only $k-2$ edges, and $\overline{S'}$ includes
no element of $\trail{k}$, either.




To find a subgraph $S$ with the desired properties, 
we further distinguish two cases according to the parity of $n$.

\paragraph{Case 2-1: $n$ is odd.}
Let $G=(V,E-E_c)$.
Then, $G$ is Eulerian since the degree of each vertex of $G$ is even and $G$ is connected.
Thus, $G$ contains a trail with $|E|-n+1 = 2k-2-n \geq k+2-n$ vertices. 
Let $T$ be a subgraph of $G$ obtained by
the first $k+1-n$ edges of such a trail, and
let $S=C\cup T$.
Then, $S$ is Eulerian or semi-Eulerian with $k+1$ edges.

\paragraph{Case 2-2: $n$ is even.}
Let $G=(V,E-E_c-\{\{v_i,v_{n/2+i}\}\mid i=1,2,\dots,n/2\})$.
Then, $G$ is Eulerian since the degree of each vertex of $G$ is even
and $G$ is connected.

Thus, $G$ contains a trail with $|E|-3n/2+1 = 2k-2-3n/2 \geq k+2-n$ vertices since $n>5$ and $n(n-1)/2 = 2k-3$. 
Let $T$ be a subgraph of $G$ obtained by the first $k+1-n$ edges of such a trail, and let $S=C\cup T$.
Then, $S$ is Eulerian or semi-Eulerian with $k+1$ edges.

\paragraph{Case 3: $|E|=2k-2$.}
This case is analogous to Case 2 where $|E|=2k-3$.
Note that for
a complete graph $K_n$ with $2k-2$ edges.
we have $n = (1+ \sqrt{1+8(2k-2)})/2=(1+\sqrt{16k-15})/2 > 5$ from Lemma \ref{lem:m-complete}.

We have to take care of the argument after constructing $S$ because
$\overline{S'}$ has $k-1$ edges and we need a different argument to
show that $\overline{S'}$ includes no element of $\trail{k}$.
Remind that $S'$ contains $k-1$ edges from a trail of $G$ and
the edges of $C$.

We distinguish two cases according to the parity of $n$.
First, let $n$ be odd.
Then, the degree of every vertex of $K_n$ is even.
Since $S'$ has four odd-degree vertices,
$\overline{S'}$ has four odd-degree vertices, too.
Thus, $\overline{S'}$ is neither Eulerian nor semi-Eulerian.
Since $\overline{S'}$ has only $k-1$ edges,
$\overline{S'}$ includes no element of $\trail{k}$.

Next, let $n$ be even.
We first observe that $n\geq 8$.
We already know that $n\geq 6$, but if
$n=6$, then the number of edges of $K_n$ is $15$, which is not of the
form $2k-2$: this is impossible.
Therefore, $\overline{S'}$ has at least four odd-degree vertices
since $S'$ has $n-4$ even-degree vertices and $n-4 \geq 4$.
Thus, $\overline{S'}$ is neither Eulerian nor semi-Eulerian.
Since $\overline{S'}$ has only $k-1$ edges,
$\overline{S'}$ includes no element of $\trail{k}$.
\end{proof}

\section{Main Theorem: Upper Bound}

We already observed that $R(\trail{k}, \trail{k}) \leq \lfloor 3k/2\rfloor -1$
as a trivial upper bound.
Now, we improve the upper bound in the next theorem.

\begin{theorem}
\label{thm:ub}
For every integer $k\geq 2$,
\[
R(\trail{k}, \trail{k}) \leq k.
\]
\end{theorem}

To this end, 
for any graph $G$ with $k$ vertices, 
we prove either $G$ or its complement $\overline{G}$ contains a trail with $k$ vertices.

We begin with the following lemma which will be used in the proof of the theorem.
\begin{lemma}
  \label{lem}
  Let $G=(V,E)$ be a bipartite graph with partite sets $A$ and $B$,
  i.e., $A\cup B=V$, $A\cap B=\emptyset$ and each edge of $G$ joins a vertex of $A$ and a vertex of $B$.
  If $|A|=3$ and the degree of every vertex of $B$ is two,
  then $G$ contains a trail such that both endpoints belong to $A$ and the
  number of edges is $2|B|$.
\end{lemma}
\begin{proof}
  Denote the three elements of $A$ by $a_1$,$a_2$, and $a_3$.
  We distinguish the following two cases according to the existence of an isolated vertex (i.e., a vertex of degree zero) in $A$.

  \paragraph{Case 1: $A$ has an isolated vertex.}
  Without loss of generality, assume that $a_3$ is an isolated vertex.
  Since each vertex in $B$ has degree two, it is adjacent to $a_1$ and $a_2$.
  Hence, the bipartite graph $G' = G - a_3$ is connected. 
  Furthermore, the number of odd-degree vertices in $G'$ is zero or two
  since the degree of $a_1$ and $a_2$ is $|B|$, and the degree of every 
  vertex in $B$ is two.
  Thus, $G'$ is Eulerian or semi-Eulerian and has $2|B|$ edges.
  When $G'$ is Eulerian, $G'$ contains a trail with $2|B|$ edges such that both endpoints coincide with $a_1$.
  When $G'$ is semi-Eulerian, $G'$ contains a trail with $2|B|$ edges such that one endpoint is $a_1$ and the other endpoint is $a_2$.

  \paragraph{Case 2: $A$ has no isolated vertex.}
  Without loss of generality, assume that there exists a vertex 
  $b\in B$ adjacent to $a_1$ and $a_2$. 
  Since there is no isolated vertex, there is a vertex $b'\in B$ 
  adjacent to $a_3$. 
  As the degree of $b'$ is two, $b'$ is adjacent to either $a_1$ or
  $a_2$. Therefore, the three vertices of $a_1$, $a_2$ and $a_3$ are
  connected by paths.
  This implies that $G$ is connected since every vertex 
  in $B$ is adjacent to one of $a_1$, $a_2$ and $a_3$. 
  Since $G$ is bipartite and the degree of every vertex in $B$ is two,
  the sum of the degrees of $a_1$, $a_2$ and $a_3$ is even.
  If there are an odd number of odd-degree vertices in $A$, then the sum of
  the degrees of $a_1$, $a_2$ and $a_3$ is odd, contradicting the fact
  that the sum of the degrees of $a_1$, $a_2$ and $a_3$ is
  even. Therefore, the number of odd-degree vertices is zero or two. 

  When there is no odd-degree vertex, then $G$ is Eulerian, and
  contains a trail with $2|B|$ edges such that both endpoints
  coincide with $a_1$.
  When there are two odd-degree vertices, let them be $a_s$ and $a_t$.
  Then, $G$ is semi-Eulerian, and contains a trail with $2|B|$ edges 
  such that one endpoint is $a_s$ and the other endpoint is $a_t$.
\end{proof}

We are now ready for the proof of Theorem \ref{thm:ub}.

\begin{proof}[Proof of Theorem \ref{thm:ub}.]
The proof uses the induction on $k$.
When $k\leq 10$, $R(\trail{k},\trail{k})$ holds from \tablename~\ref{table}.

Now, fix an arbitrary integer $k\geq 11$ and suppose that the statement is true for $k'<k$.
Consider a graph $G=(V,E)$ with $k$ vertices.
For a subgraph $G'$ with $k-1$ vertices of $G$, by induction hypothesis,
either $G'$ or $\overline{G'}$ contains a trail $S$ with $k-1$ vertices.
If $G'$ contains $S$, then $G$ contains $S$ because $G'$ is a subgraph of $G$.
If $\overline{G'}$ contains $S$, then $\overline{G}$ contains $S$ because $\overline{G'}$ is a subgraph of $\overline{G}$.
Therefore, either $G$ or $\overline{G}$ contains $S$.
Without loss of generality, suppose $G$ contains $S$.
Let $S=u_1 e_1 u_2 e_2 \dots e_{k-2} u_{k-1}$ 
where $e_i = \{u_i, u_{i+1}\}$ for all $i\in\{1,2,\dots,k-2\}$,
$U=\{u_1, u_2, \dots, u_{k-1}\}$ be the set of vertices in $S$, and $W=V-U$.
Note that the size of $U$ can be smaller than $k-1$ since some vertices can be identical.

If there exists a vertex $w\in W$ such that $\{u_1,w\}\in E$, 
then $G$ contains the trail $w\{w,u_1\}S$ with $k$ vertices. 
Similarly, if there exists a vertex $w\in W$ such that $\{u_{k-1},w\}\in E$,
then $G$ contains the trail $S\{u_{k-1},w\}w$ with $k$ vertices. 
If there is a vertex $u \in U$ such that $\{u,u_1\}\in E$ is not included in $S$, 
then $G$ contains the trail $u\{u,u_1\}S$ with $k$ vertices. 
Similarly, if there is a vertex $u \in U$ such that $\{u,u_{k-1}\}\in E$ is not included in $S$, 
then $G$ contains the trail $S\{u_{k-1},u\}u$ with $k$ vertices. 
In all of these cases, $G$ contains a trail with $k$ vertices and we are done.

Hence, we only need to consider the cases where the following two conditions are satisfied.
\begin{description}
\item[Condition 1.]
  For every $w\in W$, $\{u_1,w\}\notin E$ and $\{u_{k-1},w\}\notin E$. That is, $\{u_1,w\}\in \overline{E}$ and $\{u_{k-1},w\}\in \overline{E}$.
\item[Condition 2.]
  For every $u\in U$, if $\{u,u_1\}$ is not included $S$, then $\{u,u_1\}\notin E$. That is, $\{u,u_1\}\in \overline{E}$. If $\{u,u_{k-1}\}$ is not included $S$, then $\{u,u_{k-1}\}\notin E$. That is, $\{u,u_{k-1}\}\in \overline{E}$.
\end{description}

We distinguish the cases according to the ``shape'' of $S$.

\paragraph{Case 1: $S$ is a path.}
Since $S$ is a path, $S$ contains no repeated vertex.
Therefore, $|U|=k-1$ and $|W|=1$. 
Let $w$ be the only vertex in $W$. 
Since $S$ contains no repeated vertex, 
for $3\leq i \leq k-1$, the edges $\{u_i,u_1\}$ are not included in $S$. 
Also, for $1\leq i \leq k-3$, the edges $\{u_i,u_{k-1}\}$ are not included in $S$. 
From Condition 2, $\{u_1,u_{k-1}\}\in \overline{E}$ and 
for $3\leq i \leq k-3$, $\{u_1,u_i\}\in \overline{E}$ and $\{u_{k-1},u_i\}\in \overline{E}$. 
Consider a subgraph $G'=(V',E')$ of $\overline{G}$, 
where $V'=V-\{u_2,u_{k-2}\}$ and $E'=\{\{u_1,u_i\},\{u_{k-1},u_i\} \mid i\in\{3,4,\dots,k-3\}\} \cup \{\{u_1,w\}, \{u_{k-1},w\},\{u_1,u_{k-1}\}\}$. Each vertex of $V'$ except $u_1$ is adjacent to $u_1$. Hence, $G'$ is connected. 
Further, since the degree of each vertex in $V'$ except $u_1$ and $u_{k-1}$ is two, and the degrees of $u_1$ and $u_{k-1}$ are $|V|-3$, 
the number of odd-degree vertices in $G'$ is zero or two. 
Therefore, $G'$ is Eulerian or semi-Eulerian. 
Since $G'$ has $2k-7$ edges, $G'$ contains a trail $T$ with $2k-6\geq k$ vertices. 
Since $G'$ is a subgraph of $\overline{G}$, 
we conclude that $\overline{G}$ contains $T$.

\paragraph{Case 2: $S$ is a circuit.}
When $S$ is a circuit, $u_1 = u_{k-1}$.
Therefore, $|U|\leq k-2$ and $|W|=k-|U|\geq 2$. 
Denote the elements of $W$ by $w_1,w_2,\dots,w_{|W|}$.

If there exist $w \in W$ and $u_i \in U$ such that
$\{w, u_i\} \in E$, then
we have a trail
\[
T = w\{w,u_i\}u_i e_i u_{i+1} \dots u_{k-1} e_1 u_{2} \dots u_i
\]
since $u_1 = u_{k-1}$.
Note that $T$ has $k$ vertices.
Therefore, $G$ contains a trail with $k$ vertices.

Hence, 
we only need to consider the situation where $\{w,u\}\notin E$, i.e., $\{w,u\}\in \overline{E}$ for every $w\in W$ and every $u\in U$.
We distinguish two cases according to the comparison of $|U|$ and $|W|$.

\paragraph{Case 2-1: $|U|\geq |W|$.}
Choose two vertices $w_1,w_2 \in W$ arbitrarily, and 
let $V'=U\cup \{w_1,w_2\}$ and $E'=\{\{w_1,u\},\{w_2,u\} \mid u\in U\}\subseteq \overline{E}$.
Consider the subgraph $G'=(V',E')$ of $\overline{G}$.
Then, $G'$ is connected since every vertex in $U$ is adjacent to $w_1$ and $w_2$.
The degree of every vertex in $U$ is two, and 
the degree of $w_1$ and $w_2$ are both $|U|$.
Hence, the number of odd-degree vertices in $G'$ is zero or two.
Therefore, $G'$ is Eulerian or semi-Eulerian. 
Since $G'$ has $2|U|$ edges, 
$G'$ contains a trail $T$ with $2|U|+1\geq |U|+|W|+1 = |U|+(k-|U|)+1=k+1$ vertices. 
Since $G'$ is a subgraph of $\overline{G}$, 
we conclude that $\overline{G}$ contains $T$.

\paragraph{Case 2-2: $|U|\leq |W|$.}
If $|U|<2$, then the number of vertices in $S$ is less than 1, 
which contradicts the fact that the number of vertices in $S$ is $k-1\geq 10$.
Therefore, $|U|\geq 2$.
Choose two vertices $a,b\in U$ arbitrarily, 
and let $V'=W \cup \{a,b\}$ and $E'=\{\{w,a\},\{w,b\}\mid w\in W\}\subseteq \overline{E}$.
Consider the subgraph $G'=(V',E')$ of $\overline{G}$.
Since every vertex in $W$ is adjacent to $a$ and $b$, 
$G'$ is connected. 
The degree of every vertex in $W$ is two, and the degree of $a$ and $b$ are both $|W|$. 
Hence, the number of odd-degree vertices in $G'$ is zero or two.
Therefore, $G'$ is Eulerian or semi-Eulerian. 
Since $G'$ has $2|W|$ edges, 
$G'$ contains a trail $T$ with $2|W|+1\geq |U|+|W|+1 = |U|+(k-|U|)+1=k+1$ vertices.
Since $G'$ is a subgraph of $\overline{G}$,
we conclude that $\overline{G}$ contains $T$.

\paragraph{Case 3: $S$ is not a path or a circuit.}
Since $S$ is not a path, $|U|\leq k-2$.
Since $S$ is not a circuit, $u_1\neq u_{k-1}$.
We distinguish cases according to the size of $U$.

\paragraph{Case 3-1: $|U|=k-2$.}
Since $S$ is a trail with $k-1$ vertices and $|U|=k-2$, 
there is only one vertex $x$ that is used more than once in $S$.
If $x\neq u_1$ and $x\neq u_{k-1}$, then there are at most two vertices adjacent to either $u_1$ or $u_{k-1}$ in $G$.
If $x$ is $u_1$ or $u_{k-1}$, then there are at most four vertices adjacent to either $u_1$ or $u_{k-1}$ in $G$.

Let $U'$ be the set of elements of $U-\{u_1,u_{k-1}\}$ that are not adjacent to either $u_1$ or $u_{k-1}$ in $G$. 
Every vertex $u' \in U'$ satisfies $\{u',u_1\}\notin E$ and $\{u',u_{k-1}\}\notin E$, i.e. $\{u',u_1\}\in \overline{E}$ and $\{u',u_{k-1}\}\in \overline{E}$. 
Further, $|U'|\geq |U-\{u_1,u_{k-1}\}|-4=|U|-2-4=k-8$. 
From Condition 1, for each vertex $w\in W$, we have $\{u_1,w\}\in \overline{E}$ and $\{u_{k-1},w\}\in \overline{E}$.
Let $V'=U'\cup W \cup \{u_1,u_{k-1}\}$, $E'=\{\{w,u_1\},\{w,u_{k-1}\}\mid w\in W\}\cup \{\{u_1,u'\},\{u_{k-1},u'\} \mid u'\in U'\}\subseteq \overline{E}$ 
and consider the subgraph $G'=(V',E')$ of $\overline{G}$. 
Since every vertex in $U'\cup W$ is adjacent to $u_1$ and $u_{k-1}$, 
$G'$ is connected. 
The degree of every vertex in $U'\cup W$ is two, and 
the degree of $u_1$ and $u_{k-1}$ are both $|U'|+|W|$. 
Hence, the number of odd-degree vertices in $G'$ is zero or two. 
Therefore, $G'$ is Eulerian or semi-Eulerian. 
Since $G'$ has $2(|U'|+|W|)$ edges, 
$G'$ contains a trail $T$ with $2(|U'|+|W|)+1\geq 2(k-8+2)+1=2k-11\geq k$ vertices. 
Since $G'$ is a subgraph of $\overline{G}$, we conclude that
$\overline{G}$ contains $T$.

\paragraph{Case 3-2: $|U|\leq \left\lfloor k/2\right\rfloor$.}
From Condition 1, for each vertex $w\in W$, 
we have $\{u_1,w\}\in \overline{E}$ and $\{u_{k-1},w\}\in \overline{E}$. 
Let $V'=W \cup \{u_1,u_{k-1}\}$, $E'=\{\{w,u_1\},\{w,u_{k-1}\}\mid w\in W\}\subseteq \overline{E}$ and consider the subgraph $G'=(V',E')$ of $\overline{G}$.
Since every vertex in $W$ is adjacent to $u_1$ and $u_{k-1}$, 
$G'$ is connected. 
The degree of every vertex in $W$ is two, and 
the degree of $u_1$ and $u_{k-1}$ are both $|W|$. 
Hence, the number of odd-degree vertices in $G'$ is zero or two. 
Therefore, $G'$ is Eulerian or semi-Eulerian. 
Since $G'$ has $2|W|$ edges, $G'$ contains a trail $T$ with $2|W|+1\geq 2\left\lceil k/2\right\rceil+1\geq k+1$ vertices.
Since $G'$ is a subgraph of $\overline{G}$,
we conclude that $\overline{G}$ contains $T$.

\paragraph{Case 3-3: $\left\lfloor k/2\right\rfloor<|U|\leq k-3$.}
By the induction hypothesis, in $G$ or $\overline{G}$, 
there exists a trail $T$ with $|W|\geq 3$ vertices such that every vertex in $T$ is an element of $W$.  
Let $T=w_1 e'_1 w_2 e'_2 \dots e'_{|W|-1} w_{|W|}$ with $e'_i = \{w_i, w_{i+1}\}, i\in\{1,2,\dots,|W|-1\}$ and $W'$ be the set of vertices used in $T$.

We further distinguish two cases according to the containment of $T$ in $G$ or $\overline{G}$.

\paragraph{Case 3-3-1: $T$ is included in $G$.}
Assume that there exists a vertex $u_i\in U$ adjacent to two vertices $w_x,w_y\in W'$ where $w_x\neq w_y, x<y$.
Then, we have a trail
\[
S'=u_1 e_1 u_2 e_2 \dots u_i \{u_i,w_x\} w_x \dots w_y \{w_y,u_i\} u_i e_{i} \dots e_{k-2} u_{k-1}.
\]
The number of vertices of $S'$ is at least $k$.
Thus, we only need to consider the case where, for every vertex $u\in U$, 
there exists at most one element of $W'$ adjacent to $u$ in $G$.

Let $w'_1$, $w'_2$ and $w'_3$ be any three different vertices of $W'$.
For every vertex $u\in U$, there exists at most one element of $W'$ adjacent to $u$ in $G$. 
Then, $u$ is adjacent to at least two vertices of $w_1$,$w_2$, and $w_3$ in $\overline{G}$.
Therefore, $\overline{G}$ has the following bipartite graph $G'$ as a subgraph:
\begin{itemize}
\item The partite sets of $G'$ are $A=\{w'_1,w'_2,w'_3\}$ and $B=U$;
\item The degree of each vertex in $B$ is two.
\end{itemize}
From Lemma \ref{lem}, $G'$ has a trail $X$ with $2|B|$ edges, 
i.e., $2|B|+1=2|U|+1>2\cdot\lfloor k/2 \rfloor+1 \geq k$ vertices.
Since $G'$ is a subgraph of $\overline{G}$, $\overline{G}$ also has $X$.

\paragraph{Case 3-3-2: $T$ is included in $\overline{G}$.}
Let $w'_1$, $w'_2$ and $w'_3$ be any three different vertices of $W'$.
First, assume that in $G$ there exist three vertices $u_x,u_y,u_z\in U - \{u_1, u_{k-1}\}$ that are adjacent to at least two of the vertices $w'_1$, $w'_2$ and $w'_3$. 
Then, one of the graphs in \figurename~\ref{T} always appears as a subgraph of $G$.
In both cases, there exists a cycle $C$ that has a vertex in $U$.
Therefore, we have a trail $S'=u_1 e_1 u_2 e_2 \dots e_{x-1} C e_{x} \dots  e_{k-2} u_{k-1}$, and the number of vertices of $S'$ is at least $k$. 

\begin{figure}[t]
  \centering
  \begin{tikzpicture}[]
    \fill [white] (0,0);
    \fill [black] (-0.5,0.25)  circle (2pt);
    \fill [black] (-0.5,1.25) circle (2pt);
    \fill [black] (-0.5,2.25) circle (2pt);
    \fill [black] (1,0.25) circle (2pt);
    \fill [black] (1,1.25) circle (2pt);
    \fill [black] (1,2.25) circle (2pt);
    \draw[] (-0.5,0.25) -- (1,0.25) -- (-0.5,1.25) -- (1,2.25) -- (-0.5,2.25) -- (1,1.25)  -- cycle;        
    
    \fill [black] (3,0)+(-0.5,0.25)  circle (2pt);
    \fill [black] (3,0)+(-0.5,1.25) circle (2pt);
    \fill [black] (3,0)+(-0.5,2.25) circle (2pt);
    \fill [black] (3,0)+(1,0.25) circle (2pt);
    \fill [black] (3,0)+(1,1.25) circle (2pt);
    \fill [black] (3,0)+(1,2.25) circle (2pt);
    \draw[] (2.5,1.25) -- (4,2.25) -- (2.5,2.25) -- (4,1.25)  -- cycle;        
  \end{tikzpicture}
    
  \caption{Two subgraphs of $G$ that can be constructed by $w'_1,w'_2, w'_3, u_x, u_y$ and $u_z$.}
  \label{T}
\end{figure}
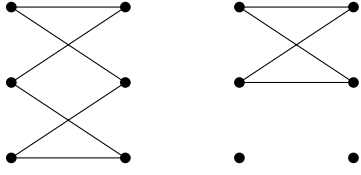

Second, assume there are at most two elements of $U-\{u_1,u_{k-1}\}$ that are adjacent to at least two vertices of $w'_1$, $w'_2$ and $w'_3$ in $G$.
Let $c$ and $d$ be those two vertices of $U-\{u_1,u_{k-1}\}$.
Then, there is a subbipartite graph $G'$ in $\overline{G}$:
\begin{itemize}
\item The partite sets of $G'$ are $A=\{w'_1,w'_2,w'_3\}$ and $B=U-\{u_1,u_{k-1},c,d\}$;
\item The degree of each vertex in $B$ is two.
\end{itemize}
From Lemma \ref{lem}, $G'$ has a trail $X$ with $2|B|+1=2(|U|-4)+1$ vertices.
Let $s,t\in W'$ be the endpoints of $X$.
We now construct a trail $T'$ with $|W|$ vertices such that it only
consists of the vertices and edges used in $T$, and does not start at $t$.
If $w_1\neq t$, then we have $T'=T$.
If $w_1=t$ and $T$ is a circuit, then we have $T'=w_2 e'_2 \dots e'_{|W|-1} w_{|W|} e'_1 w_2$ since $w_1= w_{|W|}$ and $w_1\neq w_2$.
If $w_1=t$ and $T$ is not a circuit, then we have $T'=w_{|W|} e'_{|W|-1} w_{|W|-1} e'_{|W|-2} \dots e'_{1} w_{1}$ since $w_1\neq w_{|W|}$.
Therefore, $T'$ can be constructed.

Since $T$ is included in $\overline{G}$, $T'$ is also included in $\overline{G}$.
Let $w,x$ be the endpoints of $T'$.
Then, we have a trail $Y=X\{t,u_1\}u_1\{u_1,w\}T'\{x,u_{k-1}\}u_{k-1}$ with $2|U|+|W|-5$ edges. 
Hence, $Y$ is a trail with $2|U|+|W|-4=2|U|+(k-|U|)-4=k+|U|-4> k+\lfloor k/2\rfloor-4\geq k$ vertices.
\end{proof}

\section{Conclusion}
From Theorems \ref{thm:lb} and \ref{thm:ub}, we conclude that
$R(\trail{k},\trail{k}) = k$ when $k \leq 6$ and
$2 \sqrt{k}+\Theta(1) \leq  R(\trail{k},\trail{k}) \leq k$ when
$k \geq 7$.

Future work is to find stricter upper and lower bounds.
Another challenge is to find upper and lower bounds of $R(\trail{k},\trail{\ell})$ for any $k$ and $\ell$.


\begin{thebibliography}{9}
\bibitem{GRS}
R. L. Graham, B. L. Rothschild, 
J. H. Spencer,
\emph{Ramsey Theory},
Wiley, 1990.
\bibitem{GG}
L. Gerencs\'{e}r, A. Gy\'{a}rf\'{a}s,
On Ramsey-type problems, \emph{Annales Universitatis Scientiarum Budapestinensis, E\"{o}tv\"{o}s Sect. Math.}, 10 (1967): pp.\ 167--170.
\bibitem{KR}
M. Katz, J. Reimann,
\emph{An Introduction to Ramsey Theory:
Fast Functions, Infinity, and Metamathematics}.
AMS, 2018.
\bibitem{survey}
S. Radziszowski,
Small Ramsey numbers,
\emph{The Electronic Journal of Combinatorics},
DS1, Version 16 (2021).
\end{thebibliography}
\end{document}